\documentclass{daj}

\dajAUTHORdetails{%
  title = {Boolean Functions with Small Approximate Spectral Norm}, 
  author = {Tsun-Ming Cheung, Hamed Hatami, Rosie Zhao, and Itai Zilberstein},
  plaintextauthor = {Tsun-Ming Cheung, Hamed Hatami, Rosie Zhao, Itai Zilberstein},
  keywords = {Theoretical computer science, Boolean analysis, Complexity theory},
}   

\dajEDITORdetails{%
   year={2024},
   number={6},
   received={14 November 2022},   
   published={18 September 2024},  
   doi={10.19086/da.122971},       
}   

\usepackage{amsmath,amsthm,amssymb,amscd,amstext,amsfonts}
\usepackage[utf8]{inputenc}
\usepackage{mathtools}
\usepackage{mathrsfs}
\usepackage{thmtools}
\usepackage{latexsym}
\usepackage{verbatim}
\usepackage{framed}
\usepackage{graphicx}
\usepackage{stmaryrd}
\usepackage{enumitem}
\usepackage{fullpage}
\usepackage{bm}
\usepackage{url}
\usepackage{physics}
\usepackage{dsfont}
\usepackage{todonotes}

\usepackage{color}

\usepackage[nameinlink, noabbrev, capitalize]{cleveref}

\theoremstyle{plain}
\newtheorem{conjecture}{Conjecture}

\newtheorem{thm}{Theorem}
\newtheorem{theorem}[thm]{Theorem}

\newtheorem{corollary}{Corollary}
\newtheorem{lemma}{Lemma}
\newtheorem{proposition}{Proposition}
\newtheorem{definition}{Definition}

\newtheorem{claim}{Claim}
\newtheorem*{claim*}{Claim}

\theoremstyle{remark}
\newtheorem{remark}{Remark}

\newcommand{\defeq}{\mathrel{:\mkern-0.25mu=}}

\renewcommand{\norm}[1]{\|#1\|}                     
\newcommand{\set}[1]{\{#1\}}
 
\def\supp{{\mathrm{supp}}}                               
\DeclarePairedDelimiter\ceil{\lceil}{\rceil}             

\DeclareMathOperator{\Ex}{\mathbb{E}}           
\DeclareMathOperator*{\Exs}{\mathbb{E}}           
\DeclareMathOperator{\Var}{Var}                        

\newcommand{\N}{\mathbb{N}}

\newcommand{\R}{\mathbb{R}}
\newcommand{\F}{\mathbb{F}}
\newcommand{\Z}{\mathbb{Z}}

\newcommand{\cF}{\mathcal F}

\newcommand{\cP}{\mathcal P}

\newcommand{\cX}{\mathcal X}

\newcommand{\chA}{\mathbb{A}}

\def\1{\mathbf{1}} 
\def\0{\mathbf{0}} 

\newcommand{\e}{\mathbf e}
\DeclareMathOperator{\EE}{E}
\DeclareMathOperator{\MM}{M}
\DeclareMathOperator{\tower}{Tower}

\begin{document}

\begin{frontmatter}[classification=text]

\title{Boolean Functions with Small Approximate Spectral Norm} 

\author[tmc]{Tsun-Ming Cheung}
\author[hamed]{Hamed Hatami\thanks{Supported by an NSERC grant}}
\author[rz]{Rosie Zhao}
\author[iz]{Itai Zilberstein}

\begin{abstract}
The sum of the absolute values of the Fourier coefficients of a function $f:\F_2^n \to \R$ is called the spectral norm of $f$.    Green and Sanders' quantitative version of  Cohen's idempotent theorem  states that if the spectral norm of $f:\F_2^n \to \{0,1\}$ is at most $M$, then the support of $f$ belongs to the ring of sets generated by at most $\ell(M)$ cosets,  where $\ell(M)$ is a constant that only depends on $M$.
    
We prove that the above statement can be generalized to \emph{approximate} spectral norms if and only if the support of $f$ and its complement satisfy a certain arithmetic connectivity condition. In particular, our theorem provides a new proof of the quantitative Cohen's theorem for $\F_2^n$.
\end{abstract}
\end{frontmatter}

\section{Introduction} 
Let $G = \F_2^n$ be the Boolean cube, and $\widehat{G} \cong  \F_2^n$ be its Pontryagin dual.  For a character $\chi \in \widehat{G}$, the corresponding Fourier coefficient of a function $f:G \to \R$ is defined as 
$$\widehat{f}(\chi) \defeq \Exs_{x \in G} [f(x) \chi(x)].$$

The sum of the absolute values of the Fourier coefficients is called the \emph{algebra norm} or \emph{spectral norm} of $f$, and is denoted by
$$\norm{f}_A \defeq \norm{\widehat{f}}_1 = \sum_{\chi \in \widehat{G}} |\widehat{f}(\chi)|.$$ 
The term \emph{algebra} norm is explained by the inequality $\|f g\|_A \le \|f\|_A \|g\|_A$. This norm arises naturally in theoretical computer science in connection to learning theory, and it has been studied for several complexity classes of Boolean functions~\cite{MR3620782,MR1247193,MR3246269,10.4230/LIPIcs.CCC.2021.39,10.5555/3135595.3135610,10.1145/3313276.3316319}. These studies are often motivated by the existence of efficient learning algorithms for the classes of Boolean functions that have small algebra norms~\cite{MR1247193}. Furthermore, in recent years, tail bounds in the Fourier $L_1$ norm have also become essential in constructing pseudo-random generators~\cite{v015a010,DBLP:conf/approx/ReingoldSV13,DBLP:conf/focs/ForbesK18}  and separating quantum and classical computation~\cite{10.1145/3313276.3316315,DBLP:conf/focs/Tal20,10.1145/3406325.3451040}.  

For a set $A \subseteq G$, let $\1_A$ denote the indicator function of $A$. The set of Boolean functions that satisfy $\|f\|_A=O(1)$ are fully characterized by an important theorem in harmonic analysis, Cohen's idempotent theorem. The base case of this characterization is described through the following simple proposition that characterizes the set of contractive  Boolean functions, i.e. those with $\|f\|_A \le 1$.

\begin{proposition}[Folklore, see {\cite[Proposition 1.2]{Green_Sanders}}]
\label{thm:Kwada-Ito}
A nonempty set $A \subseteq  G$ satisfies $\|\1_A\|_{A} \le 1$ if and only if $A$ is a coset of a subgroup of $G$, in which case  $\|\1_A\|_A =1$. 
\end{proposition}

It is possible to apply set operations to cosets and construct more sophisticated Boolean functions with algebra norm $O(1)$. Recall that a \emph{ring of sets} on $G$ is a subset of $\mathcal{P}(G)$ that includes $G$, and is closed under complements and (finite) intersections and (finite) unions. We say  $A \subseteq G$ has \emph{coset complexity} at most $\ell$ if it belongs to the ring of sets generated by at most $\ell$ cosets.

It is straightforward to show that  if $A$ has coset complexity at most $\ell$, then $\|\1_A\|_A \le 3^\ell$  (see \cref{lem:BoundA}). The quantitative version of Cohen's idempotent theorem states that the converse is essentially true.

\begin{theorem}[Quantitative Cohen's  theorem for $\F_2^n$ \cite{MR133397,Green_Sanders}] \label{thm:Cohen}
If $A \subseteq  G$ satisfies $\|\1_A\|_{A} \le M$, then $A$ belongs to the ring of sets generated by at most $\ell$ cosets where $\ell=\ell(M)$ depends only on $M$. 
\end{theorem}

The term ``idempotent'' essentially refers to the assumption that $f=\1_A$ is Boolean, which is equivalent to $f^2=f$. We should remark that Cohen's original theorem~\cite{MR133397} is concerned with locally compact Abelian groups of infinite size. The quantitative version  of the theorem, which is also applicable to finite groups, is due to Green and Sanders~\cite{MR2456890,Green_Sanders}. We will discuss this in more detail in \cref{sec:History}.

\paragraph{Approximate algebra norm}
Our goal is to extend \cref{thm:Cohen} to the set of Boolean functions with small \emph{approximate algebra norms}. For any error parameter $\epsilon >0$, the $\epsilon$-\emph{approximate algebra norm} of $f:G \to \R$ is defined as 
$$\norm{f}_{A,\epsilon} \defeq \inf\{\norm{g}_A \ : \ \|f-g\|_\infty \le \epsilon\}.$$ 
We remind the reader that despite what the notation might suggest, $\norm{\cdot}_{A,\epsilon}$ is not a norm. We will always assume $\epsilon \in [0,\frac{1}{2})$ as for Boolean functions, the range $\epsilon \ge \frac{1}{2}$ is trivial and uninteresting. 
 
Approximate norms, in general, are important in the theory of computation as they capture various notions of   randomized query and communication complexity. For example, approximate algebra norms are closely related  to the  notions of  randomized parity decision tree complexity and the randomized communication complexity of the so-called {\sc xor}-lifts\footnote{The {\sc xor}-lift of a function $f:\F_2^n \to \R$ is the function $f^\oplus:\F_2^n \times \F_2^n \to \R$ defined as $(x,y) \mapsto f(x+y)$.}. More precisely, the gaps between these  parameters are at most exponential, with no dependencies on the dimension $n$. We refer the reader to~\cite{MR3620782,10.1145/3471469.3471479,Hambardzumyan2021DimensionfreeBA} for more details.  

Boolean functions that have small approximate algebra norms have been studied by  M\'ela~\cite{MR665414} and Host, M\'ela, and Parreau~\cite{MR839692} under the concept of  \emph{$\epsilon$-quasi-idempotents}.  M\'{e}la \cite{MR665414} and Sanders \cite[Theorem 2.5]{MR3991375} proved that the assertion of Cohen's idempotent theorem remains true under the weaker assumption that $\|\1_A\|_{A,\epsilon} \le M$ provided that $M \le c \log \epsilon^{-1}$ for some universal constant $c>0$.\footnote{The logarithm here and throughout the paper, unless otherwise specified, is in base $2$.}   

On the other hand, Hamming balls of radius $1$ show that the requirement  $M \le c \log \epsilon^{-1}$ for some universal constant $c$ is necessary as for $B_k=\{x\in \{0,1\}^k \ : \  \sum_{i=1}^k x_i\leq 1\}$ and $0<\epsilon<\frac{1}{2}$, we have (see \cref{lem:Bk})
$$\|\1_{B_k}\|_A \ge \frac{\sqrt{k}}{2} \qquad  \text{and } \qquad \|\1_{B_k}\|_{A,\epsilon} \le 5 \log \epsilon^{-1}.$$
These bounds show that the functions $\1_{B_k}$ can have arbitrarily large coset complexity, while their approximate algebra norm is uniformly bounded by a constant that depends only on $\epsilon$. Therefore,  any potential extension of  Cohen's theorem to approximate algebra norm needs to overrule the possibility of containing  ``affine copies'' of arbitrarily large $\1_{B_k}$. This is achieved through the notion of affine connectivity.

\begin{definition}[Affine connectivity]
\label{def:connectedness}
We say that a set $A \subseteq G$ is $k$-affine  connected if    $a_0,a_0+a_1,\ldots,a_0+a_k  \in A$ implies that at least one of the following holds: 
\begin{itemize}
    \item The vectors $a_0, a_0+a_1,\ldots, a_0+a_k$ are affinely dependent;
    \item There exists  $T \subseteq \{1,\ldots,k\}$ with $|T| \ge 2$ such that $a_0+\sum_{i \in T} a_i \in A$.
\end{itemize}
\end{definition}

\begin{remark}\label{rem:connectedness}
\cref{def:connectedness} means that no restriction of $A$ to a  $k$-dimensional coset is a copy of $B_k$.  Note also that by the change of variables $b_0=a_0$ and  $b_i=a_0+a_i$ for $i=1,\ldots,k$, one can equivalently define $k$-affine  connectivity as the condition that for every linearly independent $b_0,b_1,\ldots,b_k \in A$, there exists $S\subseteq\set{0,1,\ldots,k}$ such that $\abs{S}>1$ is odd and $\sum_{i \in S} b_i \in A$.
\end{remark}

\paragraph{Our contribution}
We prove that if   $\|\1_A\|_{A,\epsilon}$ is small, then $\|\1_A\|_A$ is small if and only if both $A$ and $A^c$ are $k$-affine connected for a small $k$. 

\begin{theorem}[Main theorem]
\label{thm:main}
For every $k,M \in \N$ and $\epsilon \in [0,\frac{1}{2})$, there exists $\ell=\ell(k,M,\epsilon) \in \N$ such that the following holds. If $A \subseteq  G$ satisfies $\|\1_A\|_{A,\epsilon} \le M$, then  
\begin{enumerate}[label={\normalfont(\roman*)}]
    \item either  $A$ or $A^c$ is not $k$-affine connected, in which case $\|\1_A\|_A \ge \frac{\sqrt{k}}{2}$;
    \item or both $A$ and $A^c$ are $k$-affine connected, in which case $A$ has coset complexity at most $\ell$. In particular, $\|\1_A\|_A \le 3^\ell$.   
\end{enumerate}
\end{theorem}
\begin{remark}
Our proof results in the bound  $\ell(k,M,\epsilon) \le \tower_2\left(O(\frac{Mk}{1-2\epsilon})\right)$, where $\tower_2(m)$ denotes the tower of exponentiation with base $2$ and height $m$.
\end{remark}
\begin{remark}
    \cref{thm:main} implies \cref{thm:Cohen}: If $\|\1_A\|_A \le M$, then by \cref{thm:main}~(i) one deduces that both $A$ and $A^c$ are $O(M^2)$-affine connected, and thus one can apply \cref{thm:main}~(ii) to conclude \cref{thm:Cohen}. However,  the best known upper bound~\cite{MR3991375} for \cref{thm:Cohen} is only  $2^{O(M^{3+o(1)})}$, while this proof results in a tower-type bound.  
\end{remark}

\begin{remark}
\label{exmpl:quadratic}
The $k$-affine connectedness by itself does not imply that $\norm{\cdot}_A$ is small, and thus it is also essential that in \cref{thm:main},  we assume  $\|\1_A\|_{A,\epsilon} \le M$. For example, consider the quadratic function $f:\F_2^{n} \to \F_2$ for even $n$ defined as $f(x)=x_1x_2+x_3x_4+\ldots+x_{n-1}x_n$ where the additions are in $\F_2$. Since $f$ is a quadratic function, it satisfies 
\begin{equation}
\label{eq:Quadratic}
\sum_{S \subseteq [4]} f\left(a_0+\sum_{i \in S}a_i\right) \equiv 0.
\end{equation}
for all $a_0,\ldots,a_4 \in \F_2^n$. It follows from \cref{eq:Quadratic} that $\supp(f)$ and $\supp(f)^c$ are both $4$-affine connected. On the other hand, since $f$ is of high rank, by standard Gauss sum estimates~\cite[Lemma 3.1]{MR2773103} or a direct calculation, one can easily verify that $\|f\|_A=\Theta(2^{n/2})$. 
\end{remark}


\subsection{Historical remarks and the general picture}\label{sec:History}


\cref{thm:Kwada-Ito} is a special case of the Kawada-It\^o theorem ~\cite[Theorem 3]{MR3462}, which dates back to 1940. Kawada and It\^o characterized idempotent \emph{probability measures} on compact groups as the normalized Haar measures of compact subgroups. This theorem  was rediscovered independently by Wendel~\cite{MR67904} in the context of harmonic analysis. Later, Rudin~\cite{MR105593,MR108689}, trying to extend this result to all idempotent measures on locally compact Abelian groups,  showed
that any such measure is concentrated on a compact subgroup. Finally, Cohen~\cite{MR133397}, building on the works of Helson~\cite{Helson1953NoteOH} and Rudin~\cite{MR105593}, obtained a full description of idempotent measures on locally compact Abelian groups. Numerous extensions and refinements  of Cohen's theorem have been discovered since~\cite{MR316975, MR860817,MR2456890,MR2306837,MR2773105, MR4072207,MR4361895}.

To state Cohen's original theorem in full generality, we need a few definitions: Let $G$ be a locally compact group, and let $\widehat{G}$ be its Pontryagin dual (which is also a locally compact Abelian group). Let $\MM(G)$ be the algebra of all bounded regular Borel measures on $G$, where multiplication is defined by convolution. Let $B(\widehat{G})$ denote the  Fourier–Stieltjes algebra of $\widehat{G}$, which is the set of all $\widehat{\mu}:\widehat{G} \to \mathbb{C}$ for all $\mu \in \MM(G)$ endowed with the norm 
$\|\widehat{\mu}\|_{B(\widehat{G})} \defeq \|\mu\| $.  This norm is well-defined since the choice of $\mu$ is unique. If $\widehat{G}$ is a \emph{finite} Abelian group, then $B(\widehat{G})$ is the set of all functions on $\widehat{G}$, and  $\|\cdot\|_{B(\widehat{G})}$ coincides with the algebra norm: $\|f\|_{B(\widehat{G})}=\|f\|_{A}$.

Note that if  $\mu \in \MM(G)$ is idempotent (i.e.  $\mu*\mu=\mu$), then  $\widehat{\mu}^2 =\widehat{\mu}$, so $\widehat{\mu}(\chi) \in \{0,1\}$ for all $\chi \in \widehat{G}$. Hence the problem of characterizing all idempotent measures in $\MM(G)$ is equivalent to finding all subsets  $A \subseteq \widehat{G}$ with  $\1_A \in B(\widehat{G})$.

We say that a set $A \subseteq G$ has \emph{coset complexity} at most $\ell \in \N$ if it belongs to the ring of sets generated by at most $\ell$ \emph{open} cosets. The coset complexity of $A$ is defined to be infinite if no such $\ell$ exists.

\begin{theorem}[Cohen's idempotent theorem~\cite{MR133397}]
\label{thm:cohenGeneral}
Let $G$ be a locally compact Abelian group. A set $A \subseteq G$ satisfies $\1_A \in B(G)$ if and only if the coset complexity of $A$ is finite. 
\end{theorem}

We refer the interested readers to \cite[Chapter 3]{MR1038803} for more details. Cohen's theorem left open whether the coset complexity of $A$ is \emph{uniformly} bounded from above by a function of $\norm{\1_A}_{B(G)}$.  Moreover, it gave no information for finite groups.  Five decades later,  Green and Sanders~\cite{MR2456890,Green_Sanders}, using modern tools from additive combinatorics,  proved a stronger quantitative version of Cohen's theorem that resolved the uniformity question. Their result can be applied to finite groups as well. It states that  if $\|\1_A\|_{B(G)} \le M$, then the coset complexity of $A$ is at most $\ell(M)$, where $\ell(\cdot)$ is a universal  function that does not depend on the choice of the underlying group $G$.  They first proved the special case of this  theorem for the groups $\F_2^n$, and afterwards generalized it to all locally compact Abelian groups in \cite{MR2456890}. The bounds obtained in these two papers were later improved by Sanders~\cite{MR3991375,MR4072207}.

Cohen's theorem has been generalized to non-Abelian locally compact groups in~\cite{MR316975,MR860817}, and the quantitative version of the non-Abelian idempotent theorem was also established by Sanders~\cite{MR2773105}.

It seems conceivable that with a proper generalization of  the notion of affine connectivity,  \cref{thm:main} can similarly be generalized to  all locally compact Abelian groups, or even all locally compact groups. We defer this to future research.

\subsection{Notation}
For a positive integer $n$, we use $[n]$ to denote $\{1,\ldots,n\}$. We denote the complement of a set $S$ by $S^c$.  We will  use the standard asymptotic notations of $O(\cdot),\Omega(\cdot),\Theta(\cdot),o(\cdot),\omega(\cdot)$. Sometimes, we shall add subscripts to these notations to indicate that the constants involved depend on these parameters. For example, $O_\epsilon(1)$ means bounded from above by a constant that depends only on $\epsilon$.   

For  integers $s>0,t \ge 0$, let $\tower_s(t)$ be defined recursively as $\tower_s(t)=s^{\tower_s(t-1)}$ with the base case $\tower_s(0)=1$. For $s>1$, let $\log_s^*(m)$ be the smallest integer $t \ge 0$ such that $\tower_s(t) \ge m$.

Throughout the article, $G$ always denotes $\F_2^n$.   We consider $G$ as both a group and a vector space over $\F_2$.  
We denote by $\0 \in G$ the zero vector. For $i=1,\ldots,n$, let  $\e_i \in G$ denote the $i$-th standard vector. 

The \emph{additive energy} of a set $A \subseteq G$ is defined as 
\begin{equation}
\label{eq:additiveEn}
\EE(A) = |\{(a_1,a_2,a_3,a_4) \in A^4 \ : \  a_1+a_2=a_3+a_4\}| = |G|^3 \sum_{a \in G} |\widehat{\1_A}(a)|^4.
\end{equation}

For sets $A,B \subseteq G$ and $c \in G$, we define 
$$A+c=\{a+c \ : \ a \in A\},$$
and
$$A+B=\{a+b \ : \ a \in A, \ b \in B\}.$$

 We often identify $G \cong \widehat{G}$ via the bijection that maps $a \in G$ to the character $\chi_a(x)\defeq (-1)^{a^tx}$. We will use the notation $\widehat{f}(a) \defeq \widehat{f}(\chi_a)$. 

The \emph{convolution} of two functions $f_1,f_2:G \to \R$ is defined as 
$$f_1*f_2(x) = \Exs_{y \in G}[f_1(x-y)f_2(y)].$$
For a subgroup $W \subseteq G$, we define $\mu_W:G \to \R$ as $\mu_W:x \mapsto \frac{|G|}{|W|}\1_W(x)$ so that 
$$f*\mu_W(x) = \Exs_{y \in G}[f(x-y)\mu_W(y)]= \Exs_{y \in W+x} [f(y)]=\Exs[f|W+x].$$
The \emph{annihilator} of $W$ is defined as 
$$W^\perp = \{r \in \widehat{G} \ | \ r^ta=0 \  \text{for all} \ a \in W\}. $$
Note that convolution with $\mu_W$ corresponds to the projection of the Fourier expansion to $W^\perp$: 
\begin{equation}
\label{eq:annihilator}
f*\mu_W(x) = \sum_{a \in W^\perp} \widehat{f}(a) \chi_a(x).
\end{equation}

We call a subset of $G$ a \emph{coset} if it is a coset of some subgroup of $G$. For a subgroup $W \subseteq G$, we identify the quotient  space $G/W \equiv W^\perp$. We denote the cosets of $W$ by $W+a$, and whenever such a notation is used, it is always assumed that $W$ is a subgroup and $a \in G$.

For a function $f:G \to \R$, $a\in G$ and a subgroup $W \subseteq G$, we often identify $f|_{W+a}$ with the function on $W$, defined as $w \mapsto f(w+a)$. Note that for $w \in W$, we have 
$$f(w+a)= \sum_{b \in W}\sum_{c \in W^\perp} \widehat{f}(b+c) \chi_{b+c}(w+a)=\sum_{b \in W} \left(\chi_b(a)\sum_{c \in W^\perp} \widehat{f}(b+c)\chi_c(a)\right)\chi_b(w). $$
Hence with this notation 
\begin{equation}
\label{eq:FourierRestriction}
\widehat{f|_{W+a}}(b)=\chi_b(a) \sum_{c \in W^\perp} \widehat{f}(b+c)\chi_c(a) \qquad  \text{for all $b \in W$}.
\end{equation}

Finally, sometimes it will be more convenient to work with a slight variant of the algebra norm that excludes the principal Fourier coefficient. For a function $g:G \to \R$, define 
$$\|g\|_{\chA} \defeq \|g -\Ex[g]\|_A=\sum_{\chi \neq 0} |\widehat{g}(\chi)|.$$
For a function defined on a subgroup $W \subseteq G$, we write ${\chA(W)}$ in the subscript to emphasize the domain of the function.

\subsection{Proof overview}
\label{sec:ProofOutline}
Before giving an overview of the proof of \cref{thm:main}, we discuss Green and Sanders'   proof of \cref{thm:Cohen}. It follows a similar high-level approach as Cohen's proof~\cite{MR133397}, but uses results from additive combinatorics to obtain effective bounds. 

\paragraph{Green and Sanders' proof of \cref{thm:Cohen}} The proof uses  a strong induction which requires generalizing the statement from Boolean functions to \emph{almost integer-valued functions}.  For $\epsilon \in [0,\frac{1}{2})$, a function $f:G \to \R$ is called  $\epsilon$-integer-valued if $\|f-h\|_\infty \le \epsilon$ for an integer-valued function $h$.  

Let $f:G \to \R$ be an $\epsilon$-integer-valued function with $\|f\|_A \le M$. By \cref{eq:annihilator}, for every subgroup $W \subseteq G$, we have  
\begin{equation}
    \norm{f}_A = \norm{f*\mu_W}_A +   \norm{f-f*\mu_W}_A. \label{eq:conv-subgroup}    
\end{equation}
The main step of the proof is establishing the existence of a subgroup $W$ and a small $\delta>0$ such that 
\begin{enumerate}[label={\normalfont(\roman*)}]
\item $\|f*\mu_W - \sum_{i=1}^{c} \pm \1_{W+a_i}\|_\infty \le \epsilon+\delta$ where $c=O_{\delta,M,\epsilon}(1)$; 

\item $\|f*\mu_W\|_A \ge \frac{1}{2}$. 
\end{enumerate}
By~(i) $f*\mu_W$ is approximated by a sum that involves only $O_{\delta,M,\epsilon}(1)$ cosets. On the other hand, since $f$ and $f*\mu_W$ are $\epsilon$- and $(\epsilon+\delta)$-integer-valued respectively, their difference $f-f*\mu_W$ is $(2\epsilon+\delta)$-integer-valued. Moreover, by (ii) and \cref{eq:conv-subgroup}, we have $\norm{f-f*\mu_W}_A \le M-\frac{1}{2}$, and with this decrease in the algebra norm, we can apply the induction hypothesis to $f-f*\mu_W$ to complete the proof.  

\paragraph{The $M = O(\log \epsilon^{-1})$ requirement}  In order to reduce $M$ to $M-\frac{1}{2}$,  the ``error parameter'' is increased from $\epsilon$ to $2 \epsilon+\delta$. Repeating this process inductively for $2M$ steps will decrease the algebra norm to the base case $\|f\|_A \le \frac{1}{2}$. However, for a meaningful approximation, it has to be ensured that the error parameter never exceeds $\frac{1}{2}$.  Since the error parameter is more than doubled at each step, it is essential to require  $\epsilon=2^{-\Omega(M)}$ initially.  

The requirement that $\epsilon=2^{-\Omega(M)}$ is not just an artefact of Cohen's proof.  M\'ela~\cite{MR665414} constructed an example on a certain Abelian group which illustrates that this requirement is necessary. In \cref{lem:Bk} we show that an analogous result holds for $\F_2^n$.

\paragraph{Overview of proof of \cref{thm:main}}  \cref{thm:main} assumes that $A$ and $A^c$ are $k$-affine connected and $\|\1_A\|_{A,\epsilon} \le M$. These two assumptions  suffice to guarantee the existence of a subgroup $W$ with certain desired properties, similar to those used by Green and Sanders:
\begin{itemize}
    \item Affine connectivity implies the existence of a  coset $V+a$ such that $|V+a| \approx |A| \approx |(V+a) \cap A|$. 
    \item The assumption $\|\1_A\|_{A,\epsilon} \le M$ allows us to ``regularize'' $V$ to  a large subgroup $W \subseteq V$ such that every coset of $W$ is  either almost contained in $A$ or has almost no intersection with $A$. 
\end{itemize}

These parts of the proof closely follow Green and Sanders' proof of \cref{thm:Cohen}. 

The primary issue preventing us from further emulating the proof of \cref{thm:Cohen}  is   the fact that  $f - f*\mu_W$ is only $(2\epsilon+\delta)$-integer valued. In our case where $\epsilon$ is the algebra norm approximation parameter $\epsilon$ is taken to be $1/3$ (or any constant), we cannot afford a doubling in the error parameter.
For this reason, we depart from Cohen's approach and instead employ a completely new induction that focuses on $A$'s connectivity.  

Let us reformulate the definition of affine connectivity in a slightly different language.  Let $\cX\defeq \set{\0} \subseteq G$, and $r\defeq k+1$.  By \cref{rem:connectedness}, $A$ is $k$-affine connected if and only if for every $x_1,\ldots,x_r \in A \setminus \cX$ one of the following holds: 

\begin{enumerate}[label={\normalfont(\roman*)}]
    \item  There exists a non-empty set $S \subseteq [r]$ such that $\sum_{i \in S} x_i \in \cX$. 
           
    \item There exists a set $S \subseteq [r]$ such that $|S|>1$ is odd and $\sum_{i \in S} x_i \in A \setminus \cX$.
\end{enumerate}

The proof of \cref{thm:main} is by induction on $r$ and $M$.  Throughout the argument, $\cX$ always remains a union of $O(1)$ cosets. The coset complexity of $A \cap \cX$  can  be shown to be small by applying an induction on $M$ to  $A$'s restrictions to each individual coset in $\cX$.

The main component of the proof is to establish that it suffices to add $O(1)$ cosets to $\cX$ to reduce $r$. We present the details of this double induction in \cref{sec:mainInduction}.

\section{Basic facts} 
In this section, we state a few standard facts that will be used later in the proof of \cref{thm:main}. 
The following lemma shows that if the coset complexity of $A$ is small, then $\1_A$ can be expressed as a $\pm1$-linear combination of indicator functions of a few cosets, and thus $\|\1_A\|_A=O(1)$. 

\begin{lemma}
\label{lem:BoundA}
If $A \subseteq G$ has coset complexity at most $\ell$, then $\1_A$ can be expressed as 
\begin{equation}
\label{eq:pmSum}
\1_A = \sum_{i=1}^{t} \epsilon_i \1_{W_i+a_i}, 
\end{equation}
for cosets $W_i+a_i \subseteq G$, $\epsilon_i \in \{-1,1\}$, and $t \le 3^\ell$. In particular,  $\|\1_A\|_A \le 3^\ell$. 
\end{lemma}
\begin{proof}
Suppose  $A$ belongs to the ring of sets generated by $V_1+b_1, \ldots,V_\ell+b_\ell$. 
Each atom of this ring is of the form 
$$\bigcap_{i \in S} (V_i+b_i) \cap \bigcap_{j \in S^c} (V_j+b_j)^c, $$
for $S \subseteq [\ell]$. 
Notice that $\1_{(V_j+b_j)^c}=1-\1_{V_j+b_j}$, and the intersection of two cosets is a coset. Therefore we can express the indicator function of such an atom as a sum of $\pm1$-linear combination of indicator functions of at most $2^{|S^c|}=2^{\ell-|S|}$ cosets. Summing over all the atoms in $A$, we conclude that $\1_A$ can be expressed as such a sum with the number of terms  at most 
$$\sum_{S \subseteq [\ell]} 2^{\ell-|S|}=(1+2)^\ell=3^\ell.  $$
\end{proof}

Our next goal is to estimate the algebra norm and the approximate algebra norm of the Hamming ball of radius $1$. Our proof of the upper bound on the approximate algebra norm of $\1_{B_k}$ closely follows the argument of M\'ela~\cite{MR665414}.  

We first need to state a simple lemma from approximation theory.  The proof uses Chebyshev's classical characterization of best approximation by polynomials. Let $C([a,b])$ denote the set of all real-valued continuous functions on the interval $[a,b]$ equipped with the $L_\infty$ norm (i.e., supremum of absolute value). A $k$-dimensional subspace $V \subseteq C([a,b])$ is said to satisfy \emph{Chebyshev's condition} if every function in $V$ has at most $k-1$ \emph{distinct} zeros in  $[a,b]$ (See~\cite[Definition 2.9]{Riv90}). 

Chebyshev's classical theorem states that if $V$ satisfies Chebyshev's condition, and $S \subseteq [a,b]$ is a closed set (e.g., finite), then $v_0 \in V$ is the best $L_\infty$-approximation on $S$ of a given $f \in C(S)\setminus V$  if and only if the following holds: there exist points $x_1<\ldots<x_{k+1}$ in $S$ such that  $|f(x_i)-v_0(x_i)|=\|f-v_0\|_{L_\infty([a,b])}$ for all $i$, and $$f(x_i)-v_0(x_i) \qquad \text{for $i=1,\ldots,k+1$},$$ alternate in signs~\cite[Theorem 2.10]{Riv90}.

\begin{lemma}
\label{lem:MelaApprox}
Let $m>1$ be an integer, and let  $\eta_i\defeq \cos\left(\frac{m-i}{2m-1} \pi\right)$ for $i\in [m]$. There exists a function $\sigma:[m] \to \R$ such that
\begin{enumerate}[label={\normalfont(\roman*)}]
    \item $\sum_{i=1}^{m} \eta_i \sigma(i)=1$,
    \item $\sum_{i=1}^{m} \eta_i^{2k-1} \sigma(i)=0$ for every $2 \le k \le m$,
    \item $\sum_{i=1}^{m}  |\sigma(i)| \le 2m-1$.
\end{enumerate}
\end{lemma}
\begin{proof}
Let $V$ be the linear span of  $x^3,x^5,\ldots,x^{2m-1}$ over the reals. Every function in $V$ is an odd polynomial that has 0 as a root of multiplicity at least 3, thus the polynomial has at most   $m-2$ zeros in $(0,\infty)$. Since $\dim(V)=m-1$,  $V$ satisfies  Chebyshev's condition on any interval $[a,b]$ for $0<a<b<\infty$. 

Let $T_{2m-1}(x)=a_1x+a_3x^3+\ldots+a_{2m-1}x^{2m-1}$ be the Chebyshev polynomial of the first kind of degree $2m-1$.  Since $a_1=(-1)^m (2m-1)$ (see \cite[Section 1.1]{Riv90}), the function 
$$q(x)\defeq x+ \frac{(-1)^{m-1}}{2m-1}T_{2m-1}(x)$$
is in $V$. We claim that  $q(x)$ is the best  $L_\infty$-approximation on $S=\{\eta_1,\ldots,\eta_m\}$ of $f(x)=x$ by functions in $V$. 
By the trigonometric definition of the Chebyshev polynomial, it can be seen that $0<\eta_1<\ldots<\eta_{m}=1$ are the extrema points of $T_{2m-1}$ and the signs of $f(x)-q(x)=\frac{(-1)^{m}}{2m-1}T_{2m-1}(x)$ on these points alternate.
Hence, we can apply  Chebyshev's theorem and conclude that  $q(x)$ is the best   $L_\infty$-approximation of $f(x)=x$ on $\{\eta_1,\ldots,\eta_m\}$  by functions in $V$. Since 
$$\max_{i \in [m]}|q(\eta_i)-\eta_i| = \frac{1}{2m-1}\max_{i \in [m]}|T_{2m-1}(\eta_i)|=\frac{1}{2m-1},$$
we conclude that
\[\min_{c_3,c_5,\ldots,c_{2m-1}\in\R}~\max_{i \in [m]} \left|c_3\eta_i^3+c_5\eta_i^5+\ldots + c_{2m-1} \eta_i^{2m-1} - \eta_i\right| = \frac{1}{2m-1}.\]
Hence, by linear programming duality, the solution to the following optimization problem is $\frac{1}{2m-1}$, which after normalization yields the desired $\sigma$.
\[\begin{aligned}
    &\text{max} &&\sum_{i=1}^m \sigma(i) \eta_i\\
    &\text{subject to} && \sum^m_{i=1} \sigma(i)\eta_i^{2k-1}=0 \text{ for every }2\leq k\leq m\\
    &&& \sum^m_{i=1} \abs{\sigma(i)}\leq 1.
\end{aligned}\]
\end{proof}

In the next lemma, we show a separation in the algebra norm and the approximate algebra norm of the Hamming ball of radius $1$. We show that the algebra norm tends to infinity as $k$ grows, while the approximate algebra norm remains bounded by a constant that  depends logarithmically on $\epsilon$. 

\begin{lemma}
\label{lem:Bk}
Let $B_k \subseteq \{0,1\}^k$ be the set of all $x \in \{0,1\}^k$ with  $\sum_{i=1}^k x_i\leq 1$, and let $\epsilon \in (0,\frac{1}{2})$. We have 
$$\frac{\sqrt{k}}{2} \le \|\1_{B_k}\|_A \le \sqrt{k+1} \qquad  \text{and } \qquad \|\1_{B_k}\|_{A,\epsilon} \le 5 \log \epsilon^{-1}.$$
\end{lemma}
\begin{proof}
The upper bound $\|\1_{B_k}\|_A \le 2^{k/2}\|\1_{B_k}\|_2=\sqrt{k+1}$ is immediate from Cauchy-Schwarz inequality and Parseval's identity. To prove the lower bound on $\norm{\1_{B_k}}_A$, note that for $y \in G$, we have 
$\widehat{\1_{B_k}}(y)=2^{-k} \left(1+\sum_{i=1}^k (-1)^{y_i}\right).$
Hence 
$$\|\1_{B_k}\|_A =\Exs_{z \in \{-1,1\}^k} \left[\left|1+\sum_{i=1}^k z_i\right|\right]\ge \Exs_{z \in \{-1,1\}^k} \left[\left|\sum_{i=1}^k z_i\right|\right]  \ge \frac{1}{2} \left(\Exs_{z \in \{-1,1\}^k}\left[ \left|\sum_{i=1}^k z_i\right|^2\right]\right)^{1/2}=\frac{\sqrt{k}}{2},$$
where the first inequality uses the fact that $\sum_{i=1}^k z_i$ is a symmetric random variable, and the second inequality is an application of Khintchine's inequality.  

To prove the upper bound on $\|\1_{B_k}\|_{A,\epsilon}$, for $a \in \F_2^k$ and $s \in [-1,1]$, define  
$$0 \le \widehat{g}_s(a) \defeq 2^{-k} \left( \prod_{i=1}^k (1+s (-1)^{a_i})\right)=2^{-k} \sum_{x \in \F_2^k} s^{|x|} \chi_a(x),$$
where $|x|$ denotes the Hamming weight of $x$. Let  $g_s(x)=\sum_{a \in \F_2^k} \widehat{g}_s(a)\chi_a(x)$. It is also straightforward to verify that $g_s(x)=s^{|x|}$. By positivity of the Fourier coefficients $\widehat{g}_s(a)$, we have
\begin{equation} 
\label{eq:Norm_of_g}
\|g_s\|_A = g_s(0)=1 \qquad \text{for all $s \in [-1,1]$.}
\end{equation}
Moreover, substituting $s=\epsilon$  gives 
$$\left\|\frac{g_\epsilon-(1-\epsilon)\1_{\{\0\}}}{\epsilon} -  \1_{B_k}\right\|_\infty \le \epsilon .$$
This shows    $\|\1_{B_k}\|_{A,\epsilon} \le\ 2/{\epsilon}$, but this upper bound can be further strengthened. 

For $s\in [0,1]$, define $h_s:\F_2^k \to \R$ as  
\begin{equation}
    \label{eq:DescH}
h_s(x) \defeq \frac{g_s(x) - g_{-s}(x)}{2}=
\left\{\begin{array}{lcl}
s^{|x|} &\quad & \text{if $|x|$ is odd}\\
0 & & \text{if $|x|$ is even} 
\end{array}\right..
\end{equation}
By \cref{eq:Norm_of_g},
\begin{equation} 
\label{eq:Norm_of_h}
\|h_s\|_A \le \frac{\|g_s\|_A+\|g_{-s}\|_A}{2} = 1 \qquad \text{for all $s \in [0,1]$.}
\end{equation}
Let $\sigma:[m] \to \R$ and $\eta_i=\cos\left(\frac{m-i}{2m-1} \pi\right)$ for $i\in [m]$ be as defined in \cref{lem:MelaApprox}, where $m=\ceil{\log \epsilon^{-1}}$. Define $h:\F_2^k \to \R$ as 
$$h(x) = 2 \sum_{i=1}^m \sigma(i) h_{\frac{\eta_i}{2}}(x).$$
Note that 
\begin{itemize}
    \item If $|x|$ is even, then  by \cref{eq:DescH},  $h(x)=0$. 
    \item If $|x|=1$, then $h_s(x)=s$, and thus by  \cref{lem:MelaApprox}~(i),  $h(x)=2\sum_{i=1}^m \frac{\eta_i\sigma(i)}{2} =1.$ 
    \item If $|x|\le 2m-1$ is odd, then by   \cref{lem:MelaApprox}~(ii),  $$h(x)=2 \sum_{i=1}^m \sigma(i) \left(\frac{\eta_i}{2} \right)^{|x|}=0.$$
   \item If $|x| \ge 2m+1$ is odd, then  by \cref{lem:MelaApprox}~(iii) and the triangle inequality,
    $$h(x) = 2 \sum_{i=1}^m \sigma(i) \left(\frac{\eta_i}{2} \right)^{|x|}\le
      2^{1-|x|} \sum_{i=1}^m |\sigma(i)|  \le 2^{-2m}(2m-1) \le 2^{-m} \le  \epsilon. $$ 
\end{itemize}
Hence 
$\left\|(h+\1_{\{\0\}}) -  \1_{B_k}\right\|_\infty \le \epsilon$. Finally note that by \cref{eq:Norm_of_h}, we have 
$$\|h\|_A \le  2 \sum_{i=1}^m |\sigma(i)| \cdot \|h_{\frac{\eta_i}{2}}(x)\|_A 
\le 2 \sum_{i=1}^m |\sigma(i)| \le 2(2m-1) \le 4 \log \epsilon^{-1}. $$
We conclude that
$$\|\1_{B_k}\|_{A,\epsilon} \le  1 +  \|h \|_A \le 5 \log \epsilon^{-1}.$$
\end{proof}
\begin{remark}
The lower bound on $\norm{\1_{B_k}}_A$ is a special case of the well–known theorem of Rudin concerning lacunary series~\cite{MR0116177}. Similar bounds for more general symmetric functions were also studied in~\cite{MR3003564} and~\cite{Hambardzumyan2021DimensionfreeBA}. The fact that $\norm{\1_{B_k}}_{A,\epsilon} = O_\epsilon(1)$ was also previously known through the relation between randomized parity decision tree complexity and the approximate spectral norm~\cite[Lemmas 7 and 8]{MR3218542} and~\cite{MR2553112}. However, such proofs lead to polynomial dependencies on $\epsilon$.  The new feature of \cref{lem:Bk} is the logarithmic dependency on $\epsilon$, which is optimal up to a multiplicative constant factor. As we mentioned above, our proof closely followed the proof of~M\'{e}la~\cite[Section 7.1]{MR665414}, where he gave a similar construction for certain infinite Abelian groups (e.g., $\Z \times \Z$).   
\end{remark}

\cref{lem:Bk}, combined with \cref{lem:BoundA}, shows that the coset complexity of $\1_{B_k}$ is $\Omega(\log k)$, while its $\epsilon$-approximate algebra norm is at most $5\log \epsilon^{-1}$. This illustrates that the assertion of \cref{thm:Cohen} is not necessarily true under the weaker assumption that $\norm{\1_A}_{A,\epsilon} \le M$ if $M \ge 5 |\log(\epsilon)|$. On the other hand, recall that by~\cite[Proposition 7.1]{Green_Sanders}  the assertion of \cref{thm:Cohen} is true if $M \le c \log \epsilon^{-1}$ where $c$ is a universal constant.

\section{Proof of \cref{thm:main}}
We present the proof of \cref{thm:main} in two parts. In \cref{sec:W}, we prove the existence of a subgroup $W$ that satisfies the properties that will  be used in the main induction. The main induction  is presented in \cref{sec:mainInduction}.

\subsection{Part I: Finding a ``good'' subgroup $W$} \label{sec:W}

We first prove two lemmas (\cref{lem:LargeCoset} and \cref{lem:subspace}) that establish the existence of a coset $V+a$ such that $|V+a| \approx |A| \approx |(V+a) \cap A|$.  
\begin{lemma}
\label{lem:LargeCoset}
 Suppose $A \subseteq G$ has coset complexity at most $\ell$. Then there exists a coset $V+a \subseteq G$ such that
 $$|V+a| \ge 2^{-\ell}|A|  \qquad \text{and} \qquad   V+a \subseteq A.$$ 
\end{lemma}
\begin{proof}
The proof is by a simple induction on $\ell$. 
For the base case $\ell=1$, $A$ is either a coset or the complement of a coset. The former case is trivial. In the latter case, $\abs{A}\geq \abs{G}/2$ and thus one takes a coset of co-dimension 1 disjoint from $A$ to be the required coset $V+a$.

For $\ell>1$,  suppose that $A$ belongs to the ring generated by $V_1+b_1,\ldots,V_\ell+b_\ell$. Note that both $A \cap (V_\ell+b_\ell)$ and $A \cap (V_\ell+b_\ell)^c$ have coset complexity at most $\ell-1$, and  one of them has size larger than  $\frac{|A|}{2}$. Applying the induction hypothesis to this set completes the proof.   
\end{proof}

Recall that $\EE(A)$ denotes the additive energy of $A$.
The following lemma is essentially from \cite{MR3991375}. Its proof is based on several fundamental results in additive combinatorics. 

\begin{lemma} \label{lem:subspace}
If $A \subseteq G$ satisfies $\EE(A) \ge  \epsilon |A|^3$, then there exists a coset $V+a$  with
$$|V+a| \ge 2^{-O((\log \epsilon^{-1})^{3+o(1)})}|A| \qquad  \text{and} \qquad |A \cap (V+a)| \ge  2^{-O((\log \epsilon^{-1})^{1+o(1)})} |V+a|. $$ 
\end{lemma}
\begin{proof}
By the Balog-Szemer\'edi-Gowers theorem~\cite[Theorem 2.31]{MR2573797}  there is a  subset $A' \subseteq A$ such that
$|A'| \ge \epsilon^{O(1)} |A|$ and $|A'+A'| \le \epsilon^{-O(1)}|A'|$. Now we can apply  \cite[Proposition 2.2]{MR3991375} to conclude the existence of the desired coset $V+a$.
\end{proof}

The following lemma is an adaptation of \cite[Lemma 3.4]{Green_Sanders}. It says that if $\|f\|_A \le M$, then every subgroup $V$ can be regularized to a slightly smaller subgroup $W$ such that $f$ has small variance on all cosets of $W$.  

\begin{lemma}\label{lem:var}
Suppose $f:G \to \R$ satisfies $\norm{f}_A \le M$, and let $V \subseteq G$ be a subgroup and $\delta>0$ be a parameter. There exists a subgroup $W \subseteq V$ such that $\dim(W) \ge  \dim(V)- \frac{M}{\delta}$ and  
$\Var[f|W+c] \le \delta M$ for every $c\in W^\perp$.
\end{lemma}
\begin{proof}
If a function $g:\F_2^k \to \R$ satisfies   $\|g\|_A \le M$ and $|\widehat{g}(r)| \le \delta$ for all $r \neq 0$, then 
$$\Var[g]= \sum_{r \neq 0} |\widehat{g}(r)|^2 \le \delta \norm{g}_A \le \delta M.$$ 

By \cite[Lemma 3.4]{Green_Sanders}, there exists a subgroup $W \subseteq V$ such that $\dim(W) \ge  \dim(V)- \frac{M}{\delta}$, and  for every $r \not\in W^\perp$, 
\begin{equation}
\label{eq:SpectralSupport}
\sum_{r' \in W^\perp+r} |\widehat{f}(r')| \le \delta.
\end{equation}
For $c \in W^\perp$, define $g:W \to \R$ as $g(y) \defeq f(y+c)$, and note that by \cref{eq:FourierRestriction}, for every $r \in W$, 
$$\widehat{g}(r) =\sum_{t \in W^\perp} \widehat{f}(r+t) \chi_{t}(c).$$
In particular, we have $\|g\|_{A(W)} \le \|f\|_A \le M$, and moreover by \cref{eq:SpectralSupport}, we have $|\widehat{g}(r)|\le \delta$ for every $r \neq 0$. It follows that  $\Var[f|W+c]=\Var[g] \le \delta M$ as desired. 
\end{proof}

The following corollary is the conclusion of this section. It shows that if the assumptions of \cref{lem:LargeCoset} or \cref{lem:subspace} hold, then one can find the desired subgroup $W$ with the properties that are needed in the  proof of \cref{thm:main}. 
 
\begin{corollary}
\label{cor:main}
Let $M \ge \frac{1}{2}, \epsilon_1,\epsilon_2>0$, $\epsilon \in [0,\frac{1}{2})$, and $\delta <\min\set{1/2,\epsilon_2}$ be parameters. Suppose $A \subseteq G$, and $g:G \to \R$ satisfies $\|\1_A-g\|_\infty \le \epsilon$ and $\|g\|_A \le M$. If there exists a coset $V+a$  with 
$$|V+a| \ge \epsilon_1 |A| \qquad  \text{and} \qquad |A \cap (V+a)| \ge \epsilon_2 |V|,$$
then there exists a subgroup $W \subseteq V$ such that 
\begin{enumerate}[label={\normalfont(\roman*)}]
    \item $\Ex[\1_A|W+c] \le \delta$ or $\Ex[\1_A|W+c] \ge 1-\delta$ for every $c\in W^\perp$. 
    \item The set 
    \[\cF_W = \{c \in W^\perp:~\Ex[\1_A|W+c] \ge 1-\delta \}\]
    satisfies $1 \le |\cF_W| \le   2^{\frac{5M^2}{(1-2\epsilon)^2 \delta}}/\epsilon_1$. 
    
    \item If $\cF_W \neq W^\perp$, then $\|g|_{W+c}\|_{\chA(W)} \le \|g\|_{\chA}-\frac{1-2\epsilon- 2\delta}{2}$ for every $c\in W^\perp$.
\end{enumerate}
\end{corollary}
\begin{proof}
    By Lemma \ref{lem:var}, there exists a subgroup $W\subseteq V$ such that $\dim(W)\geq \dim(V)- \frac{4M^2}{(1-2\epsilon)^2 \delta}$ and $\Var[g|W+c]\leq  \frac{(1-2\epsilon)^2 \delta}{4} $ for every $c\in W^\perp$. We prove that $W$ is the desired subgroup.
    
     We first prove that (i) is satisfied. Let $\alpha=\Ex[g|W+c]$. If $\alpha \le \frac{1}{2}$, then since $\|\1_A-g\|_\infty \le \epsilon$, we have
     $$\Var[g|W+c]= \Exs_{x \in W+c} [|g(x) - \alpha|^2] \ge \Pr_{x \in W+c}[x \in A] (1-\epsilon-\alpha)^2 \ge \Pr_{x \in W+c}[x \in A] \left(\frac{1}{2}-\epsilon\right)^2,$$
     which shows 
    $$\Ex[\1_A|W+c] \le \left(\frac{1}{2}-\epsilon\right)^{-2}\cdot \Var[g|W+c] \le \delta.$$
    Similarly if $\alpha \ge \frac{1}{2}$, then $\Ex[1-\1_A|W+c] \le \delta.$

    For (ii), we first prove the lower bound by contradiction. Suppose the contrary that $\abs{\cF_W}=0$, then by (i), we have 
    \begin{equation}
        \abs{A\cap (W+c)}\leq \delta\abs{W}  \label{eq:supp-W}
    \end{equation}
    for every $c\in W^\perp$. By our choice of $\delta$, summing \cref{eq:supp-W} over all cosets $W+c$ in $V+a$ gives $\abs{A \cap (V+a)} \leq \delta\abs{V}<\epsilon_2\abs{V}$, which is a contradiction. For the upper bound on $\abs{\cF_W}$, as $\Ex[\1_A|W+c]\geq 1-\delta$ for any $c\in \cF_W$, we have
    \[\abs{A}\geq \sum_{c\in \cF_W} \abs{A\cap (W+c)}\geq \abs{\cF_W}\cdot (1-\delta)\abs{W},\]
    which yields the desired upper bound 
    \[\abs{\cF_W}\leq \frac{\abs{A}}{(1-\delta)\abs{W}}
    \leq \frac{1}{1-1/2}\cdot \frac{\abs{A}}{\abs{V}}\cdot \frac{\abs{V}}{\abs{W}}\leq \frac{2 \cdot 2^{\frac{4M^2}{(1-2\epsilon)^2 \delta}}}{\epsilon_1}\le \frac{2^{\frac{5M^2}{(1-2\epsilon)^2 \delta}}}{\epsilon_1},\]
    where the last inequality  uses the fact that $M \ge \delta$.
    
    To prove (iii), note that by \cref{eq:FourierRestriction}, we have 
    $$ \norm{g|_{W+c}}_{\chA(W)}= \sum_{b\in W\setminus\set{0}} \left|\sum_{r\in W^\perp} \widehat{g}(b+r)\chi_r(c)\right|.$$
    By the triangle inequality, we obtain the following inequality relating $\norm{g}_\chA$ and $\norm{g|_{W+c}}_{\chA(W)}$:
    $$\norm{g|_{W+c}}_{\chA(W)}
    \leq \sum_{b\in W\setminus\set{0}} \sum_{r\in W^\perp} \abs{\widehat{g}(b+r)}
    =\sum_{\substack{b\in W,r\in W^\perp\\(b,r)\neq (0,0)}} \abs{\widehat{g}(b+r)} -\sum_{r\in W^\perp\setminus\set{0}}\abs{\widehat{g}(r)}
    =\norm{g}_\chA-\sum_{r\in W^\perp\setminus\set{0}}\abs{\widehat{g}(r)}.$$
    It remains to show that the last sum is at least $\frac{1-2\epsilon- 2\delta}{2}$. By (i) and (ii), assuming $\cF_W \neq W^\perp$, there exist $c_1,c_2\in W^\perp$ such that $\Ex[g|W+c_1]\geq 1-\epsilon-\delta$ and $\Ex[g|W+c_2]\leq \epsilon+\delta$. Therefore, by the triangle inequality,
    \[1-2\epsilon- 2\delta\leq \Ex[g|W+c_1]-\Ex[g|W+c_2]=\sum_{r\in W^\perp} \widehat{g}(r) [\chi_r(c_1)-\chi_r(c_2)]\leq 2\sum_{r\in W^\perp\setminus\set{0}} \abs{\widehat{g}(r)}.\]
    This completes the proof of (iii).
\end{proof}

\begin{remark}
\cref{cor:main}~(iii) is one of the new ideas in the proof of \cref{thm:main}. Switching from $\|\cdot\|_A$ to $\|\cdot\|_{\chA}$ guarantees a significant decrease in the norm on every coset of $W$. 
\end{remark}

\subsection{Part II:  Induction}
\label{sec:mainInduction}
In this section, we finish the proof of \cref{thm:main} by presenting the main inductive argument, which is the principal novelty of our proof.  We start by strengthening the induction hypothesis through the following definition.

\begin{definition} 
\label{def:PropertyP}
Let $m,k \in \N$, and $0\le \epsilon<\frac{1}{2}$ be parameters. We say that $A \subseteq G$ and $g:G \to \R$ satisfy the property $\cP_{k,\epsilon}(m)$ if 
\begin{enumerate}[label={\normalfont(\roman*)}]
\item both $A$ and $A^c$ are $k$-affine connected, and
\item  $\|\1_A-g\|_\infty \le \epsilon$ and $\|g\|_{\chA} \le \left(\frac{1-2\epsilon}{4}\right)m$. 
\end{enumerate}
Moreover, if $t,r\in \N$ with $r\le k+1$, and $|A| \le \frac{|G|}{2}$, then we say that $A$ and $g$ satisfy the property $\cP'_{k,\epsilon}(m,r,t)$ if (i) and (ii) hold, and additionally  there exists $\cX=\bigcup_{i=1}^t \left(W_i+a_i\right)$ where every  $W_i+a_i$ is a coset in $G$ and the following conditions are satisfied:
\begin{enumerate}[label={\normalfont(\roman*)}]
    \setcounter{enumi}{2}
  
    \item $\|g|_{W_i+c}\|_{\chA(W_i)} \le \left(\frac{1-2\epsilon}{4}\right)(m-1)$ for every $i \in [t]$ and every $c \in G$. 
    
    \item  For every $x_1,\ldots,x_r \in A \setminus \cX$, either
    \begin{enumerate}[label={\normalfont(\alph*)}]
        \item \label{item:PropertyP_iva} there exists a set $S \subseteq [r]$ such that $|S|>1$ is odd and $\sum_{i \in S} x_i \in A \setminus \cX$; or,
        \item \label{item:PropertyP_ivb} there exists a nonempty $S \subseteq [r]$ such that $\sum_{i \in S} x_i \in \cX$. 
    \end{enumerate}
\end{enumerate}
  
\end{definition}  

\begin{remark}\label{rem:def2}
Note that if $A$ and $g$ satisfy $\cP_{k,\epsilon}(m)$ and $|A| \le \frac{|G|}{2}$, then   taking $\cX=\set{\0}$ shows that $A$ and $g$ satisfy $\cP'_{k,\epsilon}(m,k+1,1)$. Indeed, with these parameters, (iii) is trivially satisfied as $\|g|_{\{c\}}\|_{\chA(\set{\0})}=0$ for all $c \in G$, and (iv) is equivalent to the assumption that $A$ is $k$-affine connected.

Similarly, if $|A| > \frac{|G|}{2}$, then $A^c$ and $1-g$ satisfy $\cP'_{k,\epsilon}(m,k+1,1)$. 
\end{remark}

The following lemma is the core of the proof of \cref{thm:main}.

\begin{lemma}[Main lemma]
\label{lem:mainLemma}
Let $\epsilon,m,r,k,t$ be as in \cref{def:PropertyP}. 
If $A \subseteq G$ and $g:G \to \R$ satisfy $\cP_{k,\epsilon}(m)$, then the coset complexity of  $A$ is at most
   $$ \ell_{k,\epsilon}(m) \defeq    \tower_{16}\left((m-1)k+1+\log_{16}^* \left(\frac{1}{1-2\epsilon} \right)+O(1)\right).
    $$
If $A$ and $g$ satisfy $\cP'_{k,\epsilon}(m,r,t)$, then the coset complexity of $A$ is at most
\[
    \ell_{k,\epsilon}(m,r,t) \defeq
    \left\{
    \begin{array}{lcl}
    1 & & m=1, r \le k\\
    \tower_{16}\left(r+\log_{16}^* \max\left\{t,\ell_{k,\epsilon}(m-1)\right\}\right)&& \text{otherwise}.
    \end{array}\right. \]
In particular, $\ell_{k,\epsilon}(m,k+1,1)=\ell_{k,\epsilon}(m)$. 
\end{lemma}

\begin{proof}[Proof of \cref{lem:mainLemma}]
By \cref{rem:def2}, it suffices to only prove the second part of the lemma that concerns $\cP'_{k,\epsilon}(m,r,t)$.  The proof is by an induction on the two parameters $m$ and $r$. 

\paragraph{Base of induction $m=1$:} If $m=1$, then 
$\|g\|_{\chA} \le \frac{1-2\epsilon}{4}$, which implies $\|g-\Ex[g]\|_\infty \le \frac{1-2\epsilon}{4}$. Combining with $\| \1_A -g \|_\infty \leq \epsilon$, we have 
$$\| \1_A -\Ex[g] \|_\infty \le \frac{1+2\epsilon}{4}<\frac{1}{2}.$$
Since $\Ex[g]$ is a constant and $|A| \le \frac{|G|}{2}$,  we have $A=\emptyset$. Hence, $A$ has coset complexity $1$, which is at most $\ell_{m,k}(m,r,t)$ in both cases of $r \le k$ and $r=k+1$.

\paragraph{The case $r=1, m>1$:} In this case, by (iv), we have $A \subseteq \cX$, and  by (iii), we have $$\|g|_{W_i+a_i}\|_{\chA(W_i)} \le \left(\frac{1-2\epsilon}{4}\right)(m-1),$$
for every $W_i + a_i \subseteq \cX$. Hence,  for every $i \in [t]$, we can apply the induction  hypothesis to $A|_{W_i+a_i} + a_i \subseteq W_i$  and  conclude that $A|_{W_i+a_i}$ has coset complexity at most $\ell_{k,\epsilon}(m-1)$. Taking the union over all $W_i+a_i$ shows that the coset complexity of $A$ is at most $t \cdot  \ell_{k,\epsilon}(m-1)$. By the inequality $xy \le 16^{\max\set{x,y}}$, which is valid for all positive $x,y$, we have 
$$t \cdot  \ell_{k,\epsilon}(m-1)  \le  \tower_{16}\left(1+\log_{16}^* \max\left\{t,\ell_{k,\epsilon}(m-1)\right\}\right) =\ell_{k,\epsilon}(m,1,t), $$  
as desired.

\paragraph{The case $r>1, m>1$:} Consider \cref{def:PropertyP} (iv).  Since there are at most $2^r $ choices for $S$ and two choices (a) and (b),  one of the following must hold:
\begin{itemize}
   \item {\bf Case I:} There exists an odd $d \in [3,r]$ such that  
    $$\Pr_{x_1,\ldots,x_d \in A\setminus \cX}[x_1+\ldots+x_d \in A\setminus \cX] \ge \frac{1}{2^{r+1}}.$$
    
    \item {\bf Case II:} There exists a $d \in [r]$ such that 
    $$\Pr_{x_1,\ldots,x_d \in A \setminus \cX}[x_1+\ldots+x_d \in \cX] \ge \frac{1}{2^{r+1}}.$$
\end{itemize}
\begin{claim}
In Case I, there exists a coset $V+a$ such that 
\begin{equation}
\label{eq:CaseIFinal}
|V+a| \ge 2^{-k^K}|A\setminus \cX| \qquad  \text{and} \qquad |(A \setminus \cX) \cap (V+a)| \ge  2^{-k^K} |V|,
\end{equation}
where $K=O(1)$ is a universal constant. 
\end{claim}
\begin{proof}
Consider a fixation of $x_4,\ldots,x_d$ that maximizes the probability. We conclude that with $c=x_4+\ldots+x_d$, we have
$$\Pr_{x_1,x_2,x_3 \in A\setminus \cX}[x_1+x_2+x_3 \in (A\setminus \cX) + c] \ge \frac{1}{2^{r+1}},$$
which translates to
$$\Exs_{x_1,x_2,x_3 \in G}[\1_{A \setminus \cX}(x_1)\1_{A \setminus \cX}(x_2)\1_{A \setminus \cX}(x_3)\1_{A \setminus \cX}(x_1+x_2+x_3+c)] \ge \frac{1}{2^{r+1}} \left(\frac{|A \setminus \cX|}{|G|} \right)^3.$$
Substituting the Fourier transform of $\1_{A \setminus \cX}$,  we obtain 
$$\frac{1}{2^{r+1}}   \left(\frac{|A \setminus \cX|}{|G|} \right)^3 \le \sum_{a \in G} |\widehat{\1_{A\setminus \cX}}(a)|^4 \chi_a(c) \le \sum_{a \in G} |\widehat{\1_{A\setminus \cX}}(a)|^4.$$
Hence, by \cref{eq:additiveEn}, the additive energy of $A \setminus \cX$ is large: 
\begin{equation}
\EE(A \setminus \cX) \ge \frac{|A \setminus \cX|^3}{2^{r+1}}.
\end{equation}
We apply \cref{lem:subspace} with $\epsilon=2^{-r-1}\ge 2^{-k-2}$ to conclude the existence of a coset $V+a$ with 
$$|V+a| \ge 2^{-k^K}|A\setminus \cX| \qquad  \text{and} \qquad |(A \setminus \cX) \cap (V+a)| \ge  2^{-k^K} |V|,$$
where $K=O(1)$ is a universal constant. 
\end{proof}

We would like to obtain a similar statement for Case II. Unfortunately, this will require an application of the induction hypothesis, and it is the cause of the tower-type bound in our final result.  
\begin{claim} In Case II, there exists a coset $V+a$ such that
\begin{equation}
\label{eq:CaseIIFinal}
 |V| \ge \frac{2^{-\ell_{k,\epsilon}(m-1)-t}}{t 2^{r+1}} |A \setminus \cX|  \qquad \text{and} \qquad |(A \setminus \cX) \cap (V+a)| = |V|.
\end{equation} 
\end{claim}
\begin{proof}
Considering the structure of $\cX$,  there exists an $i \in [t]$ such that 
$$\Pr_{x_1,\ldots,x_d \in A \setminus \cX}[x_1+\ldots+x_d \in W_i+a_i] \ge \frac{1}{t2^{r+1}}.$$
Hence, there exists at least one choice of $x_2,\ldots,x_d \in A \setminus \cX$ such that %
$$\Pr_{x_1 \in A \setminus \cX}[x_1 \in W_i+a_i+x_2+\ldots+x_d] \ge \frac{1}{t2^{r+1}}.$$
Consequently, there exists  $c \in G$ such that  
$$\Pr_{x \in A \setminus \cX}[x \in W_i+c] \ge \frac{1}{t2^{r+1}},$$
or equivalently 
\begin{equation}
\label{eq:AinW}
|(A \setminus \cX) \cap (W_i+c) | \ge \frac{|A \setminus \cX|}{t2^{r+1}}.
\end{equation}
This by itself does not provide much information about $A\setminus \cX$ as $W_i+c$ could be much larger than $A \setminus \cX$. However, we have made progress by the decrease in $\norm{\cdot}_\chA$: by (iii), we have 
$$\|g|_{W_i+c}\|_{\chA(W_i)} \le  \left(\frac{1-2\epsilon}{4}\right)(m-1),$$
and thus we can apply the induction hypothesis to the restriction of $A$ to $W_i+c$ to describe its full structure. More precisely, the  coset complexity of $A|_{{W_i}+c}$ is at most $\ell_{k,\epsilon}(m-1)$. Since the coset complexity of $\cX$ is at most  $t$, it follows that the coset complexity of $(A\setminus \cX) \cap (W_i+c)$ is at most $\ell_{k,\epsilon}(m-1)+t$.
By applying \cref{lem:LargeCoset}, 
we find a coset $V+a \subseteq W_i+c$ such that
$$|V| \ge 2^{-\ell_{k,\epsilon}(m-1)-t} |(A \setminus \cX) \cap (W_i+c)| \quad \text{and} \quad    V+a \subseteq A \setminus \cX.$$ 
Combining with \cref{eq:AinW}, we have 
$$|V| \ge \frac{2^{-\ell_{k,\epsilon}(m-1)-t}}{t 2^{r+1}} |A \setminus \cX|  \qquad \text{and} \qquad |(A \setminus \cX) \cap (V+a)| = |V|.$$
\end{proof}
Let $\epsilon_1 \defeq 2^{-\ell_{k,\epsilon}(m-1)-t-\log t-k^K} \le  \min(2^{-k^K},\frac{2^{-\ell_{k,\epsilon}(m-1)-t}}{t 2^{r+1}} )$ and $\epsilon_2 \defeq 2^{-k^K}$ so that by \cref{eq:CaseIFinal} and \cref{eq:CaseIIFinal}, in both Case I and Case II, there exists a coset $V+a$ with 
$$|V| \ge  \epsilon_1 |A \setminus \cX|  \qquad \text{and} \qquad |(A \setminus \cX) \cap (V+a)| \geq \epsilon_2 |V|.$$
Now we are in a position to apply~\cref{cor:main} to $A \setminus \cX$.  For $\delta \defeq \min(\frac{1-2\epsilon}{8},\epsilon_2)$, by applying \cref{cor:main}, we find a subgroup $W$ such that 

\begin{itemize}
    \item $\Ex[\1_{A \setminus \cX}|W+c] \le \delta$ or $\Ex[\1_{A \setminus \cX}|W+c] \ge 1-\delta$ for every $c\in G$. 
    \item Since $\|g\|_A \le \|g\|_\chA+1 \le 2m$, the set 
    $$\cF_W = \{c \in W^\perp \  :\  \Ex[\1_{A\setminus \cX}|W+c] \ge 1-\delta \}$$
    satisfies $1 \le |\cF_W| \le  2^{\frac{20m^2}{(1-2\epsilon)^2 \delta}}/{ \epsilon_1}$. Furthermore, since $|A \setminus \cX| \le |A| \le |G|/2$,  we have $\cF_W \neq  W^\perp$.
    
    \item For every $c \in G$, we have 
    \begin{equation}
    \label{eq:corThird}
    \|g|_{W+c}\|_{\chA(W)} \le \|g\|_{\chA}-\frac{1-2\epsilon- 2\delta}{2} \le \|g\|_{\chA}-\frac{1-2\epsilon}{4}. 
    \end{equation}
\end{itemize}

Fix an arbitrary $c_0 \in \cF_W$, and let $\gamma \defeq \frac{2^{-2k}}{t}$. We will focus on $y \in W+c_0$. Recall that $\cX=\bigcup_{i=1}^t (W_i+a_i)$. For $i \in [t]$, define 
$$E_i \defeq \left\{a \in W_i^\perp :  \Pr_{y \in W+{c_0}}[y \in W_i+a_i+a] \ge \gamma \right\}.$$
Since the sets $W_i+a_i+a$ are all disjoint (for different $a \in W_i^\perp$), we have  $|E_i| \le 1/\gamma$.

Now we are ready to set up for the inductive step that will decrease $r$.  Define 
$$\cX'=\cX \cup  (W+\cF_W) \cup (W+\cF_W+c_0)  \cup \bigcup_{i=1}^{t} (W_i+E_i).$$
Note that 
$$\cX'=\bigcup_{i=1}^{t'} (W'_i + a'_i),$$ 
where $W'_i \in \{W_1,\ldots,W_t\} \cup \{W\}$ for all $i \in [t']$, and 
\begin{align}
t' &\le t + 2|\cF_W|+ \sum^t_{i=1} \abs{E_i} \nonumber \\
&\le t\left(1+\frac{1}{\gamma}\right)+\frac{2 \cdot  2^{\frac{20m^2}{(1-2\epsilon)^2 \delta}}}{ \epsilon_1} \nonumber \\
&\le t+t^2 2^{2k} + 2^{1+\frac{20 m^2  }{(1-2\epsilon)^2}\cdot \frac{8}{\epsilon_2(1-2\epsilon)}+\ell_{k,\epsilon}(m-1)+t+\log t+k^K} \nonumber\\
&\le t+t^2 2^{2k} + 2^{1+\frac{160 m^2 2^{k^K} }{(1-2\epsilon)^3}+\ell_{k,\epsilon}(m-1)+t+\log t+k^K} \nonumber\\
&\le 2^{4 \max\set{t,\ell_{k,\epsilon}(m-1)}},\label{eq:boundtprime}
\end{align} 
where we assumed that the $O(1)$ term in the definition of $\ell_{k,\epsilon}$ is chosen so that $\ell_{k,\epsilon}(m-1)$ significantly dominates all the terms that do not involve $t$. 

\begin{claim}
The pair $A$ and $g$ satisfies $\cP'_{k,\epsilon}(m,r-1,t')$ as witnessed by $\cX'$. 
\end{claim}
\begin{proof}
Conditions~(i) and (ii) of \cref{def:PropertyP} are trivially satisfied as $A$ and $g$ are not altered, and $\cX'$ is a union of $t'$ cosets. Condition (iii) is satisfied because either $W'_i \in \{W_1,\ldots,W_t\}$, or $W'_i=W$, and in the latter case   (iii) is satisfied by \cref{eq:corThird}.

It remains to verify (iv). Consider $x_1,\ldots,x_{r-1} \in A\setminus \cX' \subseteq A \setminus \cX$. For every nonempty $S \subseteq [r-1]$, let $a_S=\sum_{i \in S} x_i$.  If neither \ref{item:PropertyP_iva} nor \ref{item:PropertyP_ivb} hold, then 
\begin{alignat}{3}
\label{eq:assumptionI}
&a_S &&\not \in A \setminus \cX' &&\qquad\text{for every $S \subseteq [r-1]$ where  $|S|>1$ is odd;}  \\ 
&a_S &&\not \in \cX' && \qquad\text{for every nonempty $S \subseteq [r-1]$.}  \label{eq:assumptionII}
\end{alignat}
We will establish the existence of an $x_r \in (A\setminus \cX) \cap (W+c_0)$ such that $x_1,\ldots,x_{r-1},x_r$ violate both \ref{item:PropertyP_iva} and \ref{item:PropertyP_ivb} for $A$ and $\cX$, and thus contradict our initial assumption.  Pick $y \in W$ uniformly at random, and set $x_r=y+c_0$. We have the following: 
     
     \begin{itemize}
         \item Since $c_0 \in \cF_W$, we have  $\Pr[x_r \not\in A \setminus \cX]=\Pr_y[y+c_0 \not\in A \setminus \cX] \le \delta$.
         \item For every nonempty even-size $S \subseteq [r-1]$, since $a_S \not\in W+\cF_W+c_0 \subseteq \cX'$, we have $a_S+c_0 \not\in W+\cF_W$, and thus by the definition of $\cF_W$,
         $$\Pr_y[y+c_0+a_S \in A \setminus \cX] \le \delta.$$

         \item  For every nonempty $S$, since $a_S \not\in W_i+E_i \subseteq \cX'$, by applying the union bound over the cosets in $\cX$, we have         
         $$\Pr_y[y+c_0+a_S \in  \cX] \le t \max_{i \in [t]} \Pr_y[y+c_0+a_S \in W_i+a_i]= t \max_{i \in [t]} \Pr_y[y+c_0 \in W_i+a_i+a_S] \le t \gamma.$$
               
     \end{itemize}
     We apply the union bound to the above statements. Since $\delta + 2^{r-1} \delta + 2^{r-1} t \gamma <2^{r}2^{-k^K}+2^{r-1}2^{-2k}<1$,   with positive probability there exists a $y\in W$ such that for $x_r=y+c_0$, 
     \begin{itemize}
         \item $x_r \in A \setminus \cX$. 
         \item For every nonempty even-size $S \subseteq [r-1]$, we have  $x_r + \sum_{i \in S}  x_i  \not\in A \setminus \cX$.

         \item  For every nonempty $S \subseteq [r-1]$, $x_r + \sum_{i \in S}  x_i  \not\in  \cX$. 
     \end{itemize}     
These together with \cref{eq:assumptionI,eq:assumptionII} show that the sequence $x_1,\ldots,x_r$ violates \cref{item:PropertyP_iva,item:PropertyP_ivb} for $A$ and $\cX$, which is a contradiction.
\end{proof}

To finish the proof of \cref{lem:mainLemma}, we can apply the induction hypothesis to $A$ and $\cX'$. By \cref{eq:boundtprime}, the coset complexity of $A$ at most 
\begin{align*}
\ell_{k,\epsilon}(m,r-1,t') &\le  
\ell_{k,\epsilon}(m,r-1,16^{\max\set{t,\ell_{k,\epsilon}(m-1)}}) \\
&= \tower_{16}\left((r-1)+1+\log_{16}^* \max\left\{t,\ell_{k,\epsilon}(m-1)\right\}\right) = \ell_{k,\epsilon}(m,r,t).
\end{align*}
\end{proof}

Finally, we finish the proof of \cref{thm:main}.

\begin{proof}[Proof of \cref{thm:main}] \cref{thm:main}~(i) is an immediate corollary of \cref{lem:Bk}. To prove \cref{thm:main}~(ii), by assumption, $A$ and $A^c$ are $k$-affine connected  and there exists $g:G \to \R$ such that $\|\1_A-g\|_\infty\leq \epsilon$  and $\|g\|_\chA \le M$. Hence, $A$ and $g$ satisfy  $\cP_{k,\epsilon}\left(m\right)$ for  $m=\left\lceil \frac{4M}{1-2\epsilon}\right\rceil$ and by \cref{lem:mainLemma}, the coset complexity of $A$ is at most 
$$\ell_{k,\epsilon}(m)= \tower_2\left(O\left(\frac{Mk}{1-2\epsilon}\right)\right).$$ 
\end{proof}
\section{Concluding remarks and open problems}
We conclude the paper with some suggestions for future research:

\begin{itemize}
    \item Can \cref{thm:main} be extended to all locally compact Abelian groups, or more generally to all locally compact groups? As we discussed in \cref{sec:History}, we believe this is possible.
    
    \item Is the tower type bound in \cref{thm:main} necessary or is it an artifact of our proof? Note that for \cref{thm:Cohen}, Sanders' bound in~\cite{MR3991375} is only exponential in $O(M^{3+o(1)})$.

    \item What can be said about the structure of the sets $A \subseteq G$ that have small approximate algebra norm if we do not assume   affine connectivity? The following conjecture from~\cite{Hambardzumyan2021DimensionfreeBA}  remains open. 
    
    \begin{conjecture}
    \label{conj:LargeAffine}
    If $A \subseteq G$ satisfies $\norm{\1_A}_{A,\epsilon} \le M$, then there is a coset $V+a\subseteq G$ of codimension at most $\ell=O_{M,\epsilon}(1)$ such that $V+a \subseteq A$ or $V+a \subseteq A^c$.
    \end{conjecture}
    
    Note that by \cref{thm:main} and \cref{lem:LargeCoset}, such a coset $V+a$ exists if we further assume that $A$ and $A^c$ are $O(1)$-affine connected. 
    
    \item  \cref{thm:main}  belongs to a more general program that aims to characterize the functions that have complexity $O(1)$ in various natural communication and query models. As we discussed earlier, the approximate algebra norm,  randomized parity decision tree complexity, and  randomized communication complexity of the {\sc xor}-lift are exponentially equivalent~\cite{MR3620782,10.1145/3471469.3471479,Hambardzumyan2021DimensionfreeBA}. Therefore, \cref{thm:main} and \cref{conj:LargeAffine} are steps towards achieving such a characterization for the randomized parity decision tree model and the randomized communication complexity of the {\sc xor}-lifts.  
    
    Another possible application of \cref{thm:main} in this program is a potential characterization of the {\sc xor}-lifts with communication complexity $O(1)$ in the unbounded-error model of Paturi and Simon~\cite{paturi1986probabilistic}. We conjecture that those are precisely the {\sc xor}-lifts of the Boolean functions that have coset complexity $O(1)$. We intend to investigate this problem in future works.
    
\end{itemize}

\section*{Acknowledgments} 
The authors wish to thank the reviewers for their careful reading of the paper and constructive feedback. A preprint version of this paper is available at \url{https://eccc.weizmann.ac.il/report/2022/041/}.

\bibliographystyle{amsplain}


\begin{dajauthors}
\begin{authorinfo}[tmc]
Tsun-Ming Cheung\\
School of Computer Science, McGill University\\
Montreal, Canada\\
tsun\imagedot{}ming\imagedot{}cheung\imageat{}mail\imagedot{}mcgill\imagedot{}ca
\end{authorinfo}
\begin{authorinfo}[hamed]
Hamed Hatami\\
Associate Professor\\
School of Computer Science, McGill University\\
Montreal, Canada\\
hatami\imageat{}cs\imagedot{}mcgill\imagedot{}ca
\end{authorinfo}
\begin{authorinfo}[laci]
Rosie Zhao\\
School of Computer Science, McGill University\\
Montreal, Canada\\
rosie\imagedot{}zhao\imageat{}mail\imagedot{}mcgill\imagedot{}ca
\end{authorinfo}
\begin{authorinfo}[andy]
Itai Zilberstein\\
School of Computer Science, McGill University\\
Montreal, Canada\\
itai\imagedot{}zilberstein\imageat{}mail\imagedot{}mcgill\imagedot{}ca
\end{authorinfo}
\end{dajauthors}

\end{document}